\documentclass[12pt,oneside,reqno]{amsart}
\usepackage[margin=1in]{geometry}
\usepackage{color}
\usepackage{esint,amssymb}
\usepackage{graphicx}
\usepackage{MnSymbol}
\usepackage{mathtools}
\usepackage[colorlinks=true, pdfstartview=FitV, linkcolor=blue, citecolor=blue, urlcolor=blue,pagebackref=false]{hyperref}
\usepackage{microtype}
\usepackage{amsmath}
\usepackage{xifthen}
\usepackage{verbatim}
\definecolor{darkgreen}{rgb}{0,0.5,0}
\definecolor{darkblue}{rgb}{0,0,0.7}
\definecolor{darkred}{rgb}{0.9,0.1,0.1}

\newtheorem{theorem}{Theorem}
\newtheorem{proposition}[theorem]{Proposition}
\newtheorem{lemma}[theorem]{Lemma}

\theoremstyle{definition}
\newtheorem{remark}[theorem]{Remark}

\newtheorem{definition}[theorem]{Definition}

\newcommand{\cref}[1]{Corollary~\ref{c.#1}}

\numberwithin{equation}{section}
\numberwithin{theorem}{section}

\newcommand{\N}{\mathbb{N}}
\newcommand{\R}{\mathbb{R}}
\newcommand{\C}{\mathbb{C}}

\renewcommand{\subset}{\subseteq}

\newcommand{\test}[1][]{%
\ifthenelse{\equal{#1}{}}{omitted}{given}%
}

\newcommand{\pa}{\partial}

\newcommand{\EBE}{\operatorname{EBE}}

\renewcommand{\bar}{\overline}

\renewcommand{\part}{\partial}

\usepackage{dsfont}

\begin{document}

\title{The Extended Bogomolny Equations on $\mathbb {R}^2 \times \mathbb{R}^+$ with nilpotent Higgs field}

\begin{abstract}
We study the extended Bogomolny equations with gauge group $SU(2)$ on $\R^2\times\R^+$ with generalized Nahm pole boundary conditions and nilpotent Higgs 
field. We completely classify solutions by relating them to certain holomorphic data through a Kobayashi-Hitchin correspondence.
\end{abstract}

\author[P. Dimakis]{Panagiotis Dimakis}
\address[Panagiotis Dimakis]{D\'epartement de Math\'ematiques, Universit\'e du Qu\'ebec \`a Montr\'eal}
\email{pdimakis12345@gmail.com}
\subjclass[2010]{}
\date{\today}

\maketitle
\section{Introduction}

Since their conception in \cite{KW}, the Kapustin-Witten (KW) family of equations has become the subject of a number of conjectures, most prominently relating to the 
theory of knots and the Jones polynomial. The family of equations is defined on a four manifold $M^4$ equipped with an $SU(2)$-bundle $E$. 
In terms of a complex connection $\mathcal A := A + i\Phi$ on the bundle $E$, these equations can be succinctly written as 
\begin{equation*}
e^{i\theta}F_{\mathcal A} = \bar{\star e^{i\theta}F_{\mathcal A}},
\end{equation*}
where $F_{\mathcal A}:= d\mathcal A + \mathcal A \wedge\mathcal A$ is the curvature of the complex connection. Restricting to the special angle $\theta = \pi/4$, 
we recover the original Kapustin Witten equations as described in \cite{KW}. This particular case will be the subject of this paper, and we shall simply call these
{\it the} KW equations.

We state the core conjecture in this area. The goal is to count the number of solutions to the KW equations on $\R^3\times\R^+$ which satisfy a certain singular boundary
condition with respect to a knot $K$ embedded in $\R^3 \times \{0\}$. More specifically, we seek solutions which are asymptotic to certain model solutions along
$\R^3\times \{0\}$: away from $K$, the model solution is the charge zero singular monopole solution, depends only on the vertical variable and exhibits a `Nahm 
pole' at the boundary (up to a gauge transformation); near $K$, the model solution is a charge $1$ singular monopole solution, again up to a gauge transformation,
and will be recalled later. This asymptotic behavior, as well as the asymptotic behavior near infinity, make it possible to define the second relative Chern number of $E$,
and for each such Chern number $m$, the conjecture states that the number of solutions to the KW equations with these boundary conditions 
equals the $m^{\mathrm{th}}$ coefficient of the Jones polynomial. 

While the Fredholm theory for this set up has been successfully carried out in \cite{MW2}, thus making the counting problem a reasonable one, there 
is currently no way to produce solutions for any $K$, except in the simplest setting where $K$ is a straight line in $\R^3$. In \cite{GW} the authors suggested 
an adiabatic approach to this problem.  One key step is to dimensionally reduce as follows. Take a slice $\R^2\times\R^+ \subset \R^3\times\R+$ transverse to $K$ 
and stretch $\R^3 \times \R^+$ in directions perpendicular to the slice. In the limit, the knot reduces to a set of parallel lines $\{p_j\} \times \R$, where 
$\mathcal P:= \{p_1,...,p_n\}$ is the set of points where the knot intersects the slice.  The fields then satisfy Nahm pole boundary conditions at the boundary
of the slice away from $\mathcal P$ and are asymptotic to the model knot solutions around each $p_j$. These dimensionally reduced equations on $\R^2 \times \R^+$
are known as the extended Bogomolny equations (EB).   Sufficient knowledge about the moduli space of solutions to these should lead to a correspondence 
between paths in this moduli space and solutions of the KW equations themselves, using a family of parallel slices as above. (Note that we are sidestepping 
the potentially difficult analysis needed to understand slices which are not transverse to $K$)

The aim of this paper is the study solutions to the dimensionally reduced KW equations on $\R^2\times \R^+$ with generalized Nahm pole 
boundary asymptotics along the boundary $\R^2\times\{0\}$.  We focus on the case where the Higgs field $\Phi$ is nilpotent. This extends \cite{MH1, MH2}
and also extends and simplifies \cite{TaubesDC}. We highlight the role of the holomorphic data in determining the positions and charges of the knot points. 
Along the way we introduce and develop the necessary techniques needed here, as well as in a subsequent paper \cite{DM1} which analyzes the case where 
\begin{equation*}
\varphi = \begin{pmatrix} a & 0 \\ 0 & -a \end{pmatrix}
\end{equation*}
and uses such solutions to construct solutions to the KW equations on $\R^3\times\R^+$ with $2n$ braiding strands. 

While not the most general semisimple Higgs field, this is sufficient in order to obtain solutions which have both positive and negative magnetic charge with 
net magnetic charge zero. This charge cancellation appears to be essential in our construction of solutions to the KW equations in \cite{DM1}. The nilpotent case does not allow such solutions. We hope to describe the partial compactifiction of the moduli space of solutions 
corresponding to points of opposite magnetic charge being annihilated and explain how the study of this space may lead to knot invariants.

The paper \cite{MH2} constructs solutions to the extended Bogomolny equations for arbitrary holomorphic Higgs field on $\Sigma \times \R^+$, where
$\Sigma$ is a compact Riemann genus $g\geq 2$. A Kobayashi-Hitchin correspondence is established there using a Hermitian geometric 
approach. We follow this paper closely, but generalize it by replacing the cross-section $\Sigma$ with $\R^2$.  When this cross-section is compact,
there is a clear choice of how to prescribe the limits of solutions as the $\R^+$ variable tends to infinity. This is less clear when the cross-section
is noncompact, and this is the first main obstacle to overcome. In addition, the noncompactness of $\R^2$ requires more delicate a priori estimates. 

Writing the Higgs field as 
\begin{equation*}
\Phi = \Phi_1\,dx_1 + \Phi_2\,dx_2 + \Phi_3\,dx_3 +\Phi_4\,dx_4,
\end{equation*}
it turns out that the extended Bogomolny equations force $\varphi_z := \Phi_2 - i \Phi_3$ to be holomorphic.  It is natural to consider only solutions of polynomial
growth, hence the components of $\varphi_z$ must be polynomials.   We shall further restrict attention here to the case where $\varphi_z$ is a nilpotent matrix. Therefore, up to a gauge transformation, we have that 
\begin{equation}\label{simple holomorphic field form}
\varphi_z = \begin{pmatrix} 0 & P(z) \\ 0 & 0 \end{pmatrix},
\end{equation}
for some polynomial $P(z)$. 

The existence of solutions with nilpotent Higgs fields and a description of their moduli as a set was conjectured explicitly in Section 2.3 of \cite{GW}. In \cite{TaubesDC}, the author 
proved the existence of such solutions in the case where there is exactly one positively charged point at the boundary $\R^2\times\{0\}$, and named these 
KW instantons. Honoring this terminology, we will call the solutions produced in this paper KW multi-instantons.  The total charge of a multi-instanton is the
sum of charges of each of the individual charges of points in $\R^2$. Our method of proof is different from that of Taubes. We use the Hermitian-Einstein formalism 
as in \cite{MH2}, which makes the proof shorter, easier to generalize and perhaps conceptually simpler. 

We now define the two sets that we eventually prove are in bijective correspondence. 

\begin{definition}
Denote by $\EBE(N, K)$ the set of solutions to the extended Bogomolny equations on $\R^2\times \R^+$ with generalized Nahm pole boundary asymptotics
and with the total charge $K$, and such that as $\rho:= \sqrt{|z|^2 + y^2} \to \infty$, the solution is asymptotic to the model solution of charge $N$, where 
$K-N \in \{0, 2, 4, \ldots \}$. 
\end{definition}

For the next definition, we note that as explained carefully in Section \ref{holomorphic line sub-bundle}, following \cite{GW}, the asymptotic conditions at the 
boundary determine a holomorphic line subbundle $\mathcal L$ of $E$. The multi-instanton behavior of solutions is closely related to the interaction between 
$\mathcal L$ and $\varphi_z(\mathcal L)$.  
\begin{definition}
Consider the set of Hermitian-Einstein data $\{(\mathcal E, P(z), \mathcal L, K)\}$, where each four-tuple is consists of: the trivial holomorphic rank $2$ bundle on $\R^2$, 
a polynomial $P(z)$ of degree $N$ corresponding to \eqref{simple holomorphic field form}, and a holomorphic line bundle $\mathcal L$ given by a section of the form 
\begin{equation*}
s = \begin{pmatrix} Q(z) \\ R(z) \end{pmatrix}
\end{equation*}
where $Q$ and $R$ are relatively prime polynomials such that $\text{deg}(Q) \le \text{deg}(R) - 1$ and $\text{deg}(R) = (K-N)/2$.
\end{definition}
\begin{theorem}\label{main theorem}
There is a bijective correspondence between $\EBE(N,K)$ and $\{(\mathcal E, P(z), \mathcal L, K)\}$. 
\end{theorem}

As the authors correctly predict in \cite{GW}, there are some inherent limitations to extracting topological information when the underlying Higgs field is nilpotent. In the 
adiabatic approximation described above, one would like to be able to model actual knots and not just braids. Unfortunately, forcing the field to be nilpotent does not allow 
the existence of points of the same electric and opposite magnetic charge and therefore does not predict annihilation of points, or in the four dimensional picture braids
 closing up to produce a knot. Thankfully, there is a different set up which does predict the existence of such pairs of points which goes by the name of complex symmetry
 breaking. This will be the main theme of the subsequent paper \cite{DM1}. 

The analysis of the EB equations with nilpotent Higgs field requires some new analysis because of the non-compactness of the cross-section $\R^2$. In particular, 
the construction of approximate solutions requires a detailed understanding of the asymptotics of solution near the boundary and at radial infinity. The continuity 
method can then proceed using these as well as techniques from \cite{Melliptic} and better decay estimates in the radial direction. These insights are all 
essential in this paper and its further developments.  This work also illuminates how the knot points, as well as their electric and magnetic charges, are 
reflected in the holomorphic data, and in the interaction of the line subbundle and the field $\varphi$, cf.\ Subsection \ref{holomorphic line sub-bundle}. 

One important question which we pursue elsewhere is to understand what happens when $Q(z)$ and $R(z)$ vary and develop common roots. This monopole 
bubbling phenomenon seems to be of importance in the geometric Langlands program as described in \cite{KW} Section $10.2$. The behavior of solutions 
when the roots of $Q(z)$ and $R(z)$ collide leads naturally to the construction of a natural partial compactification of the moduli space $\EBE(N,K)$. 
We will pursue the analogous question in the setting of complex symmetry breaking since in that case this bubbling corresponds to points of opposite 
magnetic charge annihilating each other. 

Finally, the existence of model knot solutions for the KW family of equations with $\theta\neq \pi/4$ in \cite{PD3} makes it possible to generalize the theorems of this paper to arbitrary $\theta \neq 0,\pi/2$. Some interesting predictions about the set theoretic moduli spaces of solutions to such equations can be found in \cite[Section 3]{GW}.

Here is a guide to this paper. Section $2$ contains the framework and more detailed formulation of the problem; Subsection \ref{model solutions} provides a 
short derivation of Witten's model knot solutions \cite{WFivebranes}. Subsection \ref{holomorphic line sub-bundle} is the conceptual core of the paper,
and it is explained there how the multi-instanton behavior arises from the interaction of the holomorphic line sub-bundle with the nilpotent field $\varphi_z$. 
Approximate solutions to the equations are constructed in Section $3$. This requires identifying a good approximation to the solution as $\rho\to \infty$,
which is one of the main differences between the work here and what was done in \cite{MH2}, where the cross-section is assumed to be compact. 
Section $4$ reviews the necessary linear analysis, as developed in \cite{MH2}. Section $5$ is the technical core of the paper. Writing $H = H_0e^s$ 
where $H_0$ is the approximate metric constructed in section $3$, we consider a one-parameter family of equations $N_t(s) = 0$, where $N_0(s) = 0$ 
is the equation we want to solve, whereas $N_1(s) = 0$ has a trivial solution.  The equations $N_t(s) = 0$ are solved by continuity. The openness part 
requires better decay estimates at radial infinity, see Subsection \ref{estimates}. Closedness follows \cite{MH2} fairly closely. We can thus solve $N_0(s) = 0$. 
Uniqueness is proved in Section $6$ using a Donaldson type functional and the estimates \ref{estimates}. This uniqueness shows finally that Taubes' model 
instanton solutions are the same as the solutions coming from the holomorphic data $(P(z), Q(z), R(z)) = (z^k, a_0 + ... + a_{p-1}z^{p-1}, z^p)$ with $a_0 \neq 0$,
as explained in remark \ref{Taubes' solutions}.

\subsection{Acknowledgements} I would like to thank my advisor Rafe Mazzeo for encouraging me to continue working on this paper after the appearance of \cite{TaubesDC} and for his insightful contributions to this project across countless meetings.

\section{Setting up the problem}

\subsection{The Extended Bogomolny Equations}
We begin by reviewing the definition of the KW equations. These are defined on a four-dimensional manifold $M^4$ equipped with an $SU(2)$-bundle $E$. 
If $A$ is a connection on $E$ and $\Phi$ is an ad$(E)$-valued $1$-form, then these equations take the form
\begin{equation}\label{KW}
\begin{split}
F_A - \Phi\wedge\Phi + \star\,d_A\Phi &= 0\\
\,d_A \star \Phi &= 0.
\end{split}
\end{equation}
As in the introduction, we are interested in the case where $M^4 = \R_{x_1}\times \Sigma_z\times \R_y^+$, where $\Sigma$ is a Riemann surface. We focus
on the case where $\Sigma = \mathbb C$. We use both the holomorphic coordinate $z$ and the Cartesian coordinates $x_2,x_3$, where $z = x_2 + ix_3$.

We wish to dimensionally reduce the equations \eqref{KW} by considering solutions which are invariant in the direction of $x_1$.  We will assume that 
$A_1 = \phi_y = 0$.  In fact, by a gauge choice, it is easy to ensure that $A_1 = 0$. On the other hand, $\phi_y$ cannot be gauged away, but it turns out 
that this vanishing may be deduced from the generalized Nahm pole boundary conditions at $y = 0$ and the asymptotic conditions imposed as $y \to \infty$,
see \cite{HeGl}.

Returning to the reduction, write $\Phi = \phi + \phi_1\,dx_1$ and fix the orientation $\,dx_1\wedge\,dz\wedge\,d\bar z \wedge\,dy$; the dimensionally reduced KW equations then take the form 
\begin{equation}
\begin{split}
F_A - \phi\wedge\phi  &= \star\,d_A\phi_1\\
\,d_A\phi + \star[\phi,\phi_1] &= 0\\
\,d_A^{\star}\phi &= 0,
\end{split}
\end{equation}
where $\star$ is the Hodge star on $\Sigma\times\R_y^+$. 

Although it is not evident in this form, the dimensionally reduced KW equations possess a Hermitian Yang-Mills structure as first observed in \cite{GW} and \cite{WFivebranes}. To see this, define the operators $\Theta:= (\mathcal D_1,\mathcal D_2, \mathcal D_3)$ where
\begin{equation}
\begin{split}
\mathcal D_1 &= (D_2 + iD_3)\,d\bar z = (\pa_{x_2} + i\pa_{x_3} + [A_2+iA_3, \cdot ])\,d\bar z\\
\mathcal D_2 &= [\phi_2 - i\phi_3, \cdot ]\,dz\\
\mathcal D_3 &= D_y - i[\phi_1, \cdot ] = \pa_y + [A_y - i\phi_1, \cdot ]
\end{split}
\end{equation}
and fix a Hermitian metric $H$ on the restriction of the bundle $E$ on $\Sigma\times\R_y^+$. Since we have assumed that the bundle possesses an $SU(2)$ structure, there exists a unique such metric such that the connection and field coefficients are unitary with respect to it. Then the dimensionally reduced KW equations can be written in the form
\begin{equation}\label{EBE}
\begin{split}
[\mathcal D_i,\mathcal D_j] = 0,\, i,j = 1,\,2,\,&3,\\
\frac{i}{2} \Lambda\left( [\mathcal D_1 , \mathcal D_1^{\dagger_H}] + [\mathcal D_2 , \mathcal D_2^{\dagger_H}]\right) &+ [\mathcal D_3 , \mathcal D_3^{\dagger_H}] = 0,
\end{split}
\end{equation}
where $\Lambda:\Omega^{1,1} \to \Omega^0$ is the inner product with the K\"ahler form $\frac{i}{2}\,dz\wedge\,d\bar z$. We call this system of equations the Extended Bogomolny Equations (EBE) in accordance with \cite{MH1, MH2}. The first set of equations in \eqref{EBE} can be thought of as a complex moment map equation and the last equation as the accompanying real moment map equation. We use the following notation for the real moment map:
\begin{equation}\label{unitary moment}
\mathcal M(H) := \frac{i}{2} \Lambda\left( [\mathcal D_1 , \mathcal D_1^{\dagger_H}] + [\mathcal D_2 , \mathcal D_2^{\dagger_H}]\right) + [\mathcal D_3 , \mathcal D_3^{\dagger_H}].
\end{equation}

Let $\mathcal G_{\C}: = SL(2, E)$ the group of special linear automorphisms of the bundle $E$. To understand the Hermitian Yang-Mills structure of the equations \eqref{EBE} we study how gauge transformations $g\in\mathcal G_{\C}$ act on them. 

\begin{lemma}\label{lemma 1}
Assume that the data $(E,\Theta, H)$ satisfy \eqref{EBE}. Then the data $(E,\Theta^g, H^g)$ do as well, where $\Theta^g := (g^{-1}\circ \mathcal D_1 \circ g, g^{-1}\circ \mathcal D_2 \circ g, g^{-1}\circ \mathcal D_3 \circ g)$, and $H^g := g^{\dagger}Hg$.
\end{lemma}
\begin{proof}
The operators $\Theta^{g}$ clearly satisfy the first three equations in \eqref{EBE}. To verify that the forth equation holds, we need to compute the adjoints of the modified operators with respect to the modified metric $H^{g}$. It holds that 
\begin{align*}
\pa_{\bar z} H^{g}(s,s') &= H(gs,gs') \\
&= H(\mathcal D_1^{\dagger}gs,gs') + H(gs,\mathcal D_1gs')\\
&= H^{g}(g^{-1}\circ \mathcal D_1^{\dagger} \circ g s, s') + H^{g}(s, g^{-1}\circ \mathcal D_1 \circ g s'),
\end{align*}
and therefore, $\mathcal D_1^{\dagger_{H^{g}}} = g^{-1}\circ \mathcal D_1^{\dagger} \circ g$. Similarly, $\mathcal D_2^{\dagger_{H^{g}}} = g^{-1}\circ \mathcal D_2^{\dagger} \circ g$ and $\mathcal D_3^{\dagger_{H^{g}}} = g^{-1}\circ \mathcal D_3^{\dagger} \circ g$. Plugging these formulas into the last equation we obtain the last equation in  \eqref{EBE} conjugated with $g$, hence the last equation is also satisfied. 

\end{proof}

We have shown that the first three equations in \eqref{EBE} enjoy a larger symmetry, while the last equation is preserved only by those gauge transformations which satisfy $H = g^{\dagger}Hg$, namely the stabilizer subgroup $\mathcal G_{\C}^H$. This structure suggests an alternative approach to constructing solutions to these equations. As stated above, we are given a bundle with a fixed metric $H$ and the unknowns are a unitary triple $(A,\phi,\phi_1)$. One can instead fix a system of operators $\Theta$ which satisfy the first three equations in \eqref{EBE} and then the fourth equation becomes an equation with the metric $H$ as an unknown. 

Let us elaborate on this approach. First we describe the form the operators in $\Theta$ can have. If $f$ is a function on $\Sigma\times\R^+$ and $s$ a section of $E$, then necessarily
\begin{equation}
\begin{split}
\mathcal D_1(fs) &= (\bar \pa f)s\,d\bar z + f\mathcal D_1s\\
\mathcal D_3(fs) &= ( \pa_y f)s\,dy + f\mathcal D_3s\\
\mathcal D_2 &= [\varphi, \cdot], \textrm{ for some } \varphi \in \Omega^{1,0}(\text{Ad}(E));
\end{split}
\end{equation}
and furthermore satisfy the equations $[\mathcal D_i,\mathcal D_j] = 0$. Then, their adjoints with respect to a metric $H$ are given by the formulas 
\begin{align*}
\pa_{\bar z} H(s,s') &= H(\mathcal D_1^{\dagger} s, s') + H(s, \mathcal D_1s') \\
\pa_z H(s,s') &= H(\mathcal D_2^{\dagger} s, s') + H(s, \mathcal D_2s') \\
\pa_y H(s,s') &= H(\mathcal D_3^{\dagger} s, s') + H(s, \mathcal D_3s').
\end{align*}
Here, by definition $H(s,s'):= \bar s^TH s'$, so the inner product is complex linear in the second argument. 

From this point on, let us specialize to the case $\Sigma \cong\, \C$ where there is a particularly nice choice for the system of operators $\Theta$, which is given by $\mathcal D_1 = \bar \pa$ and $\mathcal D_3 = \pa_y$. Since we are working over $\C\times\R^+$, $\varphi = \varphi_z\,dz$ where $\varphi_z$ is a trace free matrix with holomorphic entries. With this choice of operators the last equation in \eqref{EBE} becomes 
\begin{equation}\label{moment map}
M(H) := -\bar\pa(H^{-1}\pa H) - \pa_y(H^{-1}\pa_y H) + [\varphi_z, H^{-1} \varphi_z^{\star}H] = 0,
\end{equation}
where $\varphi_z^{\star}$ is the Hermitian adjoint of $\varphi_z$. Assuming the solution $H$ to \eqref{moment map}, the immediate question is how do we extract the unitary triple $(A,\phi,\phi_1)$. The answer is to express the metric as $H = g^{\dagger}g$, and gauge transform by $g^{-1}$ so that $\Theta = (g\circ \bar \pa \circ g^{-1}, g\circ \varphi \circ g^{-1}, g\circ \pa_y \circ g^{-1})$ and their adjoints are given by $\Theta^{\dagger} = ( (g^{\dagger})^{-1}\circ \pa \circ g^{\dagger}, (g^{\dagger})^{-1}\circ \varphi^{\star} \circ g^{\dagger}, (g^{\dagger})^{-1}\circ \pa_y \circ g^{\dagger})$. From these formulas we see for example that $A_{\bar z} = -(\bar \pa g)g^{-1}$ and $A_z = (g^{\dagger})^{-1}(\pa g^{\dagger})$ so that $A_{\bar z}^{\dagger} = -A_z$, which is the necessary and sufficient condition for $A_2$ and $A_3$ to be unitary matrices. 

\subsection{Boundary and asymptotic conditions}\label{model solutions}

To describe the boundary conditions on \eqref{EBE}, we need certain model solutions. Any solution to \eqref{EBE} should locally look like one these basic solutions as $y\to 0$. We also require that any solution satisfies that the triple $(A,\phi,\phi_1)$ decays uniformly at least as fast as $\rho^{-1}$ as $\rho \to \infty$ where $\rho:= \sqrt{r^2+y^2}= \sqrt{x_2^2 + x_3^2 + y^2}$.

We start with analyzing the field $\varphi = \varphi_z\,dz$. Working in parallel holomorphic gauge, the equations $[\mathcal D_1, \mathcal D_2] = [\mathcal D_3,\mathcal D_2] = 0$ imply that the matrix $\varphi_z$ has holomorphic entries. In this paper we are interested in the case where this matrix is nilpotent. This implies that there exists a holomorphic gauge transformation $g$ that takes this matrix to 
\begin{equation}\label{Jordan canonical}
\varphi_z = \begin{pmatrix} 0 & P(z) \\ 0 & 0 \end{pmatrix},
\end{equation}
where $P(z)$ is holomorphic. We restrict to holomorphic functions that have at most polynomial growth as $r \to \infty$, which by Liouville's theorem must then be polynomials. For the model solutions we let $P(z) = z^k$ for $k\in\N$. 

To solve \eqref{moment map} we also need to choose an ansatz for $H$. The simplest possible ansatz is given by 
\begin{equation}
H_k = \begin{pmatrix} e^{-u_k} & 0 \\ 0 & e^{u_k} \end{pmatrix}.
\end{equation}
Since it is natural to expect that a solution with $P(z) = z^k$ be rotationally symmetric around the axis $z = 0$, we assume that $u_k$ is a function of $r$ and $y$. The unitary triple corresponding to this data is given by 
\begin{equation}\label{unitary triple}
\begin{split}
A_k &= r\pa_r u_k(r,y)\begin{pmatrix} \frac{i}{2} & 0 \\ 0 & -\frac{i}{2} \end{pmatrix}\,d\theta\\
\phi_{z,k} &= e^{-u_k}\varphi_z = e^{-u_k}z^k\begin{pmatrix} 0 & 1 \\ 0 & 0 \end{pmatrix}\\
\phi_{1,k} &= \pa_y u_k(r,y)\begin{pmatrix} \frac{i}{2} & 0 \\ 0 & -\frac{i}{2} \end{pmatrix}.
\end{split}
\end{equation}

Inserting these $\varphi_z$ and $H$ into \eqref{moment map}, we obtain that
\begin{equation}\label{model equation}
\left(\frac{\pa^2}{\pa x_2^2}+\frac{\pa^2}{\pa x_3^2}+\frac{\pa^2}{\pa y^2}\right)u_k + r^{2k}e^{-2u_k} = 0.
\end{equation}
We look for solutions $u_k$ which are regular on $\C\times(0,\infty)_y$ and which blow up in a specific way as $y\to 0$ away from $r = 0$. Let us elaborate on this. Consider the special case where $k = 0$. In this case we consider solutions which only depend on $y$ and hence need to solve the equation
\begin{equation}
u_0'' + e^{-2u_0} = 0.
\end{equation}
Solutions are given by 
\begin{equation}\label{model y}
\begin{split}
u_0(y) &= \text{log}\left(\frac{\sinh(b(y+c))}{b}\right)\text{~when~}b\neq0\\
u_0(y) &= \text{log}(y+c)\text{~when~}b = 0.
\end{split}
\end{equation}
Since we want the solutions which are singular as $y\to 0$, we must set $c = 0$. Also, the condition that $\phi_1\to 0$ as $y\to \infty$, implies $u_0(y) = \text{log}(y)$.

\begin{remark}
The solutions corresponding to $b\neq 0$ force $\phi_1$ to tend to a non-zero traceless matrix as $y\to \infty$. Solutions with this asymptotic condition are called real symmetry breaking solutions. Unless $k = 0$, the existence of model solutions of this form has not been rigorously proven. For interesting predictions and the possible geometric significance of these solutions see section $2.4$ in \cite{GW}. 
\end{remark}

To treat the case $k >0$, we must specify the asymptotic conditions for the model solution. Away from $z = 0$, the field 
\begin{equation}
\varphi_z = \begin{pmatrix} 0 & z^k \\ 0 & 0 \end{pmatrix} 
\end{equation}
is gauge equivalent to the field corresponding to $k = 0$ through the holomorphic gauge transformation 
\begin{equation}
g = \begin{pmatrix} e^{-ik\theta/2} & 0 \\ 0 & e^{ik\theta/2} \end{pmatrix}.
\end{equation}
Thus, as $y\to 0$ away from $z = 0$, we want $u_k$ to be gauge equivalent to $u_0$ and hence require the asymptotic condition that $u_k \sim \text{log}(y)$ as $y\to 0$. We also require that $u_k$ is smooth on $\C\times (0,\infty)$. To find an explicit solution, define 
\begin{equation}
v_k = u_k - (k+1)\text{log}(r). 
\end{equation}
The point of this transformation is to make the differential equation \eqref{model equation} homogeneous of order $-2$. Indeed equation \eqref{model equation} takes the form 
\begin{equation}
\Delta v_k + r^{-2}e^{-2v_k} = 0.
\end{equation}
This equation is scale invariant so it is reasonable to look for solutions which respect this symmetry. We thus consider $v_k$ as a function of $\sigma = y/r$. Equation \eqref{model equation} can be further reduced to 
\begin{equation}
(\sqrt{\sigma^2+1}\,\pa_{\sigma})^2 v_k + e^{-2v_k} = 0.
\end{equation}
Under the change of variables $\sigma = \sinh(\tau)$ the equation becomes
\begin{equation}
v_k'' + e^{-2v_k} = 0.
\end{equation}
Solutions to this equation which are singular along $y = 0$ are given by \eqref{model y} with $y$ being replaced by $\tau$. However, since we require $u_k$ to be regular away from the boundary, $v_k$ must be asymptotic to $-(k+1)\text{log}(r)$ along the positive $z$-axis. This condition is satisfied iff $b = k+1$ in which case the solution takes the form 
\begin{equation}
v_k = \text{log}\left(\frac{\sinh((k+1)\tau)}{k+1}\right).
\end{equation}
Going back to $u_k$ and the variable $\sigma$, we get the model solution 
\begin{equation}
e^{u_k} =\frac{(\sqrt{r^2+y^2} +y)^{k+1} - (\sqrt{r^2 + y^2} - y)^{k+1}}{2(k+1)}.
\end{equation}

We are in a position to describe the boundary conditions for a general solution to \eqref{EBE}. We give definitions both for the unitary triple $(A,\phi,\phi_1)$ and for the metric $H$. In the following, we use the variable $\psi = \tan^{-1}(r/y)$.

\begin{definition}
The unitary triple $(A,\phi,\phi_1)$ satisfies the Nahm pole boundary condition with knot singularity of charge $k$ at $(p,0)\in \C\times \R^+$ if, in some gauge, 
\begin{equation}
(A,\phi_z,\phi_1) = (A_k,\phi_{z,k},\phi_{1,k}) + \mathcal O(\rho^{-1+\epsilon}\psi^{-1+\epsilon})
\end{equation}
for some $\epsilon >0$.

We say that a hermitian metric $H$ satisfies the Nahm pole boundary condition with knot singularity of charge $k$ at $(p,0)\in \C\times \R^+$ if there exists a section $s\in i\mathfrak{su}(E,H_k)$ such that $H = H_ke^s$ and $|s| + |y\,ds| \le C\rho^{\epsilon}\psi^{\epsilon}$ for some $\epsilon>0$.
\end{definition}

\subsection{The holomorphic line sub-bundle}\label{holomorphic line sub-bundle}

Suppose that the triple $(A,\phi_z,\phi_1)$ satisfies the Extended Bogomolny Equations \eqref{EBE} as well as the boundary and asymptotic conditions described above. We will also require that the set of points $(p,0)\in \C\times \R^+$ at which the solution is asymptotic to a positive charge model solution be finite. By definition, near each point on the boundary there exists a homolorphic frame in which the solution looks like one of the model solutions; in particular,
\begin{equation}
\phi_1 = \frac{1}{2y}\begin{pmatrix} i & 0 \\ 0 & -i \end{pmatrix} + \mathcal O(\rho^{-1+\epsilon}\psi^{-1+\epsilon}).
\end{equation}
Let $s_1$ and $s_2$ be local holomorphic coordinates for the bundle $E$ so that $s_1$ corresponds to $(1,0)^{\dagger}$ and $s_2$ to $(0,1)^{\dagger}$. Consider a generic section $s = a_1s_1 + a_2s_2$ of the bundle $E|_{y=1}$ and parallel transport it using the operator $\mathcal D_3$. Then the equation $\mathcal D_3 s = \pa_y s - i\phi_1 s = 0$ implies that as $y\to 0$ 
\begin{equation}
s(y) = (a_1y^{-1/2} + \mathcal O(y^{-1/2+\epsilon}))s_1 + (a_2y^{1/2} + \mathcal O(y^{1/2+\epsilon}))s_2.
\end{equation}
Thus the section $s_2$ is special since its parallel transport vanishes like $y^{1/2}$ as $y\to 0$, whereas the parallel transport of the generic section actually blows up like $y^{-1/2}$. This is called the small section and we obtain an invariant description of the vanishing line bundle 
\begin{equation}
\mathcal L := \{ s\in \Gamma(E) : \mathcal D_3 s = 0, \lim\limits_{y\to 0} |y^{-1/2+a}s| = 0\}
\end{equation}
for all $0<a<1$. Since $\mathcal L$ is locally spanned by the holomorphic section $s_2$ it is holomorphic. Given this holomorphic line bundle and $\varphi_z$, one can easily determine the points of positive charge and the charge of each point. We explain this now. In parallel holomorphic gauge, we have the extra freedom due to the nilpotency of $\varphi_z$ to gauge it into the form \eqref{Jordan canonical} for some polynomial function $P(z)$. However, we cannot at the same time put the small section of the bundle $E$ into the form $(0,1)^{\dagger}$. This "incompatibility" is the sourse of all of the instanton behavior observed by Taubes \cite{TaubesDC}. The small section in general will have the form 
\begin{equation}\label{holomorphic section}
s = \begin{pmatrix} Q(z) \\ R(z) \end{pmatrix},
\end{equation}
where $Q(z)$ and $R(z)$ are relatively prime polynomials such that $\text{deg}(Q) \le \text{deg}(R) -1$. We can choose them to satisfy this relation because $\varphi_z$ is invariant under upper triangular transformations and therefore we can always modify $Q(z)$ by a polynomial multiple of $R$. To determine the points and their multiplicities, consider the quantity 
\begin{equation}
s\wedge \varphi_z s = P(z)R(z)^2,
\end{equation}
where $P(z)$ is the polynomial showing up in $\varphi_z$. The zeroes of this expression are the points of positive charge and the mutliplicities of the zeroes give the respective charges. This also explains the reason why the entries of the small section are polynomials as opposed to more general functions since we only consider solutions with finitely many positively charged points. 

Going back to unitary gauge, the small section transforms as $s^U = gs$ where $g$ is such that the solution metric $H$ is equal to $g^{\dagger}g$. Since the small section $gs$ is required to vanish at least as fast as $y^{1/2}$ as $y\to 0$ in unitary gauge, we can rephrase the boundary condition by requiring that $gs$ has this vanishing property for $s$ as in \eqref{holomorphic section}.

What we have shown above is that given a solution to the Extended Bogomolny Equations \eqref{EBE} which satisfies the appropriate boundary and asymptotic conditions, we can extract a field $\varphi_z$ of the form \eqref{Jordan canonical} and a holomorphic section $s$ of the bundle $E$ of the form \eqref{holomorphic section}. It is then natural to ask whether for every such pair $(E, \varphi_z,\, s)$ we can find a solution corresponding to it. This is the main theorem of this paper:
\begin{theorem}\label{Main}
Given any three polynomials $P(z),\, Q(z), \, R(z)$ such that $P$ and $Q$ are co-prime and $\text{deg}(P)\le \text{deg}(Q) - 1$ there exists a solution to the Extended Bogomolny Equations whose positively charged points are specified by the zeroes of the polynomial $P(z)R(z)^2$ and the respective charges are specified by the multiplicities of the zeroes. 
\end{theorem}

This theorem will be proved in the next three sections.

\section{The approximate solution}

In this section, given a triple $(P(z),\, Q(z),\, R(z))$ we construct a metric $H_0$ such that 
\begin{equation}
\begin{split}
\mathcal M(H_0) &= \mathcal O(\rho^{-4}) \text{ uniformly as } \rho \to \infty\\
\mathcal M(H_0) &= \mathcal O(y^{-1}) \text{ uniformly as } y\to 0\\
\mathcal M(H_0) & \text{ stays bounded everywhere else}.
\end{split}
\end{equation}
Recall that the relationship between the real moment map \eqref{unitary moment} in unitary gauge and in holomorphic gauge is given by $\mathcal M(H) = g\circ M(H) \circ g^{-1}$ where $H = g^{\dagger}g$. 

\begin{remark}
Once we develop the linear analysis required in the next section, we will be able to modify $H_0$ so that the second condition becomes $\mathcal M(H_0) = \mathcal O(y^{N}) \text{ as } y\to 0$ for all $N>0$. 
\end{remark}

There are three regions that we study separately. Define

\begin{equation}
r_0 := \{\sup |z| : P(z) = 0 \text{ or } R(z) = 0\},
\end{equation}
then consider a finite collection of disjoint open half-balls $\{U_i\}$ in $\C\times \R^+$ each centered at one of the positively charged points $p_i$. We can choose these half-balls to be small enough so that their union is strictly inside the region $\rho < 16r_0$. Then the first region will be the union of these half-balls. If we denote the radius of $U_i$ by $r_i$, then define the second region to be the open set that is the set of points which are at distance at least $r_i/2$ from $p_i$ for all $i$ and inside the half-ball $\rho < 16r_0$. Finally, the last region is the open set given by $\rho > 8r_0$. 

Let us start by constructing an approximate solution on one of the half-balls of the first region. The charge of the center point will have the form $k + 2p$ where $P(z)$ vanishes to order $k$ at this point and $R(z)$ to order $p$. Working in parallel holomorphic gauge, we still have the freedom to perform holomorphic gauge transformations. Since the polynomials $Q(z)$ and $R(z)$ are relatively prime, there exist polynomials $S(z)$ and $T(z)$ such that $Q(z)S(z) + R(z)T(z) = 1$. Then, consider the gauge transformation 
\begin{equation}
g = \begin{pmatrix} T & Q \\ -S & R \end{pmatrix}. 
\end{equation}
Under this, the small section goes to $g^{-1}s = (0,1)^{\dagger}$ and 
\begin{equation}
\varphi_z^g:= g^{-1} \varphi_z g = \begin{pmatrix} pRS & pR^2 \\ -pS^2 & pRS \end{pmatrix}.
\end{equation}
With a further diagonal holomorphic transformation we can make the upper right entry of $\varphi_z^g$ to be exactly equal to $z^{k + 2p}$. The salient features of this reduction are that there exists a gauge inside the half-ball such that the field takes the form 
\begin{equation}
\varphi_z^g = \begin{pmatrix} a & z^{k+2p} \\ b & c \end{pmatrix},
\end{equation}
where $a,\,b,\,c$ denote bounded holomorphic functions, and the small section has the particularly nice form $s = (0,1)^{\dagger}$. As we now show, in this gauge the model metric $H_{k+2p}$ produces an error that is uniformly of order $\mathcal O(y^{-1})$ throughout the half-ball and writing $H_{k+2p} = g_{k+2p}^2$, the section $g_{k+2p}s$ indeed vanishes like $y^{1/2}$ as $y\to 0$. The second claim is obvious and therefore we only need to check the first one. Plugging $H_{k+2p}$ and $\varphi_z^g$ into the moment map equation \eqref{moment map} and conjugating with $g_{k+2p}^{-1}$ it is straightforward to see that the error term is of order $\mathcal O(y^{-1})$.

For the second region the exact same argument we have applied above goes through and gives an error of order $\mathcal O(y^{-1})$. In fact in the second region one can do better. Indeed there exists a parallel holomorphic gauge in which 
\begin{equation}
\varphi_z = \begin{pmatrix} 0 & 1 \\ 0 & 0 \end{pmatrix}.
\end{equation}
In this gauge, write the small section as $s= (f_1(z), f_2(z))^{\dagger}$ where the $f_i$ are holomorphic and since we are in the second region, $f_2 \neq 0$. Therefore, we can gauge transform by 
\begin{equation}
g = \begin{pmatrix}1 & f_2^{-1}f_1 \\ 0 &1 \end{pmatrix}
\end{equation}
to make the small section equal to $(0,1)^{\dagger}$. Since this gauge transform is holomorphic and does not affect $\varphi_z$, the metric $H_0$ is an exact solution in this gauge. 

Finally, let us analyze the last region. In order to be able to produce a decent approximate solution we will need the following result which was proven in \cite{MH1}:
\begin{theorem}\label{diagonal metric solution}
Given the data $(P(z),\,0,\,1)$, there exists a solution to the Extended Bogomolny Equations corresponding to these data such that the solution metric $H$ is diagonal:
\begin{equation}
H = \begin{pmatrix} e^{-u} & 0 \\ 0 & e^u \end{pmatrix}.
\end{equation}
Furthermore, as $\rho \to \infty$, $u = (N+1)\log\rho + \log\sin\psi +\mathcal O(1)$ uniformly in $(\psi,\theta)\in S_+^2$, where $N := \deg P$.
\end{theorem}
Let $\chi$ be a smooth function on $[0,\infty)_x$ which is identically zero for $x\le 1/4$ and identically $1$ for $x\ge 3/4$. Consider the approximate solution metric given by 
\begin{equation}
\begin{split}
H_{\rho > 8r_0} &= \begin{pmatrix} e^{-u} & -e^{-u}\Sigma \\ -e^{-u}\bar \Sigma & e^{u} + e^{-u}|\Sigma|^2 \end{pmatrix}\\
\Sigma  &= \frac{\chi(r/\rho)Q(z)}{R(z)},
\end{split}
\end{equation}
where $u$ is the solution given by Theorem \ref{diagonal metric solution}. Notice that $\Sigma$ is defined everywhere since $R(z)$ has no zeroes in the region where $\chi(r/\rho)$ is positive. Writing this metric as $H_{\rho > 8r_0} = g^{\dagger}g$, we observe that 
\begin{equation}
g = \begin{pmatrix} e^{-u/2} & -e^{-u/2}\Sigma \\ 0 & e^{u/2} \end{pmatrix}
\end{equation}
satisfies the small section condition, meaning that $gs \sim \mathcal O(y^{1/2})$ as $y\to 0$. In particular this implies that the condition $\mathcal M(H_{\rho > 8r_0}) = \mathcal O(1)$ is met. Finally we need to understand the error term when we plug $H_{\rho > 8r_0}$ into the moment map equation. After some calculation it turns out that 
\begin{equation}
\mathcal M(H_{\rho > 8r_0}) = \begin{pmatrix} A & \bar B \\ B & -A \end{pmatrix},
\end{equation}
where 
\begin{equation}\label{error terms}
\begin{split}
A &= \Delta u + |P(z)|^2 e^{-2u} + 4e^{-2u}|\pa_z\bar\Sigma|^2 + e^{-2u}|\pa_y\bar\Sigma|^2\\
B &= e^{-u}(\Delta\bar \Sigma - 8\pa_{\bar z}u \cdot\pa_z\bar\Sigma - 2\pa_yu\cdot\pa_y\bar\Sigma).
\end{split}
\end{equation}
By Theorem \eqref{diagonal metric solution}, 
\begin{equation}
\Delta u  + |P(z)|^2e^{-2u} = 0
\end{equation}
and due to its asymptotics, $e^{-2u} = \mathcal O(\rho^{-2})$ in the region where $\chi$ is not identically equal to $1$. Now, in the subregion where $\Sigma$ is nonzero, it decays at least as fast as $\rho^{-1}$ due to the degree condition for the polynomials $Q$ and $R$. Furthermore, if any of the derivatives $\pa_z,\,\pa_{\bar z},\,\pa_y$ hits $\chi(r/\rho)$ it produces a factor of $\rho^{-1}$. Finally, in the region where $\chi$ is identically equal to $1$, both $\pa_z\bar\Sigma$ and $\pa_y\bar\Sigma$ are identically zero. From these observations it becomes clear that both $A$ and $B$ are uniformly of order $\mathcal O(\rho^{-4})$ as $\rho\to\infty$. 

Finally, in order to get an approximate metric on the whole $\C\times\R^+$, we glue the approximate metrics in each region using suitable rescalings of the bump function $\chi$. It is clear that the final metric will satisfy the conditions we requested in the beginning of this section.

\section{Linear Analysis}

The subject of this section is to study the mapping properties of the linearization of \eqref{unitary moment} at a given metric $H$ between appropriate Banach spaces. Large parts of this section follow the equivalent section in \cite{MH2}, but are repeated to make the exposition self-contained.

\subsection{The linearized operator}

We start by stating without proof a Proposition in \cite{MH2}. Equation \eqref{inner product identity} is needed for the a priori estimates and the linearized operator as it appears in \eqref{linearized operator} is needed for the linear analysis of this section. 
\begin{proposition}
If we write $H = H_0e^s$ for some $s\in i\mathfrak{su}(E,H_0)$ then 
\begin{equation}\label{linearization}
\mathcal M(H) = \mathcal M(H_0) + \gamma(-s)\mathcal L_{H_0}s + Q(s),
\end{equation}
where 
\begin{equation}
\begin{split}
\mathcal L_{H_0} s &:= \frac{i}{2}\Lambda(\mathcal D_1\mathcal D_1^{\dagger_{H_0}} + \mathcal D_2\mathcal D_2^{\dagger_{H_0}})s +\mathcal D_3\mathcal D_3^{\dagger_{H_0}}s \\
Q(s)s&:= \frac{i}{2}\Lambda((\mathcal D_1\gamma(-s))\mathcal D_1^{\dagger_{H_0}} + (\mathcal D_2\gamma(-s))\mathcal D_2^{\dagger_{H_0}})s + (\mathcal D_3\gamma(-s))\mathcal D_3^{\dagger_{H_0}}s,
\end{split}
\end{equation}
and $\gamma(s) := \frac{e^{\text{ad}_s}-1}{\text{ad}_s}$. Moreover, 
\begin{equation}\label{inner product identity}
\langle \mathcal M(H) - \mathcal M(H_0), s \rangle_{H_0} = \Delta |s|_{H_0}^2 + \frac{1}{2}\sum|v(s)\nabla_i s|_{H_0}^2
\end{equation}
where $v(s) = \sqrt{\gamma(-s)}$, $\Delta = \nabla^{\star}\nabla$, $\nabla_i = \mathcal D_i + \mathcal D_i^{\dagger_{H_0}}$ for $i\in\{1,2\}$ and $|v(s)\nabla_3 s|^2 = |v(s)\mathcal D_3 s|^2 + |v(s)\mathcal D_3^{\dagger} s|^2$. Finally, letting $\star$ denote the adjoint with respect to the usual inner product on forms, we have the simplified formula for the linearized operator 
\begin{equation}\label{linearized operator}
\mathcal L_H = \frac{1}{4}(\nabla_1^{\star}\nabla_1 + \nabla_2^{\star}\nabla_2) - (\nabla_y^2 - \phi_1^2) + \frac{1}{2}[\mathcal M(H),\cdot ],
\end{equation}
where $\phi_1^2 := [\phi_1,[\phi_1, \cdot]]$.
\end{proposition}

\subsection{Mapping properties of $\mathcal L_H$}

We want to prove that the linearized operator $\mathcal L_H$ is Fredholm and in fact invertible between appropriately chosen Banach spaces. Fredholmness of this operator is proven through a parametrix construction, which crucially depends on the invertibility of the associated normal operator at each point of the boundary and near $\rho \to \infty$. The normal operator $N(\mathcal L)$ is the part of $\mathcal L$ which is precisely homogeneous of degree $-2$ with respect to dilations. At a boundary point of zero charge the normal operator takes the form 
\begin{equation}
N(\mathcal L) = \Delta_{\R^3} + N_S
\end{equation}
while at a boundary point of charge $k$ or as $\rho\to\infty$ the normal operator takes the form 
\begin{equation}
N(\mathcal L) = \pa^2_{\rho} + \frac{1}{\rho}\pa_{\rho} + \frac{1}{\rho^2}\left( \pa^2_{\psi} + \cot \psi \pa_{\psi} - \frac{1}{\sin^2\psi} \nabla^{\star}_{\theta}\nabla_{\theta}\right) + N_S
\end{equation}
where in both formulas 
\begin{equation}
N_S := -\phi_1^2 + \frac{1}{2}([\phi_z^{\dagger},[\phi_z,\cdot]]+ [\phi_z,[\phi_z^{\dagger},\cdot]]).
\end{equation}
A few comments are in order. When we work near a zero charge point, the fields that appear in the normal operator are precisely the fields appearing in the model solution with zero charge. This is because we have assumed that this is the form of the solution at this point to leading order and up to gauge equivalence. Similarly, near a charge $k$ point, the fields $(A,\phi_z,\phi_1)$ that appear form the unitary triple corresponding to the model knot solution of charge $k$. Finally, near infinity the fields that appear are those of a model knot solution of charge $N:= \deg p$. It is straightforward from the construction we gave in the previous section that the approximate solution we constructed looks like this to leading order. 

To describe the spaces between which the normal operators and the linearized operator itself are invertible, we need to study the indicial roots of the normal operators. These can be considered as the formal rates $\lambda$ of vanishing or blow up of a solution near that point. In the case of the normal operator at a zero charge point, this amounts to the existence of a locally defined smooth section $s$ so that the equation $N(\mathcal L)(y^{\lambda} s) = 0$ holds. This is equivalent to requiring that $\mathcal L(y^{\lambda} s ) = \mathcal O(y^{\lambda -1})$ as opposed to $\mathcal O(y^{\lambda -2})$ which is the expected order for a generic value of $\lambda$. Similarly, in the case of the normal operator near a charged point or near radial infinity, we are looking for the $\lambda$ such that there exists a field $\Phi(\psi,\theta)$ on the half sphere $S_+^2$ such that $N(\mathcal L)(\rho^{\lambda} \Phi) = 0$. The calculations of indicial roots for the zero charge points was carried out in \cite{MW1} and for general charge in \cite{MW2}. These computations can be summarized as follows:
\begin{theorem}
The set of indicial roots at a zero charge point is given by $\{-1,2\}$. The set of indicial roots of $N(\mathcal L)$ at a charge $k$ point for any $k$ is given by the formula
\begin{equation}
-\frac{1}{2} \pm \sqrt{\gamma +1/4}
\end{equation}
where $\gamma$ is an eigenvalue of the spherical part of the operator $\mathcal N(\mathcal L)$. The smallest such $\gamma$ satisfies $\gamma_0 >2$ . 
\end{theorem}
\begin{remark}
The above Theorem makes sense precisely because the spherical part of $\mathcal N(\mathcal L)$ has discrete spectrum as proven in \cite{MW2}.
\end{remark}

Finally, before we state the main theorem of this section, we define suitable Banach spaces.
\begin{definition}
Let $C^{k,a}_{\text{ie}}(\C\times \R^+)$ be the space of functions on $\C\times \R^+$ such that 
\begin{itemize}
\item in a neighborhood $U$ of a point of zero charge 
\begin{equation}
\|u\|_{L^{\infty}} + \sup\limits_{i + |\beta| \le k} [(y\pa_y)^i(y\pa_x)^{\beta}u]_{\text{ie};0,a} <\infty
\end{equation}
where 
\begin{equation}
[u]_{\text{ie};0,a} := \sup\limits_{\substack{ (x,y)\neq (x',y')\\ (x,y),(x',y')\in U}} \frac{|u(x,y)- u(x',y')|(y+y')^a}{|y-y'|^a + |x-x'|^a},
\end{equation}
\item in a neighborhood V of a point of positive charge 
\begin{equation}
\|u\|_{L^{\infty}} + \sup\limits_{i + p + q \le k} [(\rho\pa_{\rho})^i(\pa_{\theta})^p\pa_{\psi}^q u]_{\text{ie};0,a} <\infty
\end{equation}
where 
\begin{equation}
[u]_{\text{ie};0,a} := \sup\limits_{\substack{ (\rho,\psi,\theta)\neq (\rho',\psi',\theta')\\ (\rho,\psi,\theta),(\rho',\psi',\theta')\in V}} \frac{|u(\rho,\psi,\theta)- u(\rho',\psi',\theta')|(\rho+\rho')^a}{(|\theta -\theta'| + |\phi - \psi'|)^a(\rho+\rho')^a + |\rho-\rho'|^a},
\end{equation}
\item away from the boundary $y = 0$ we require that $u$ lies in the regular H\"older space $C^{k,a}$. 
\end{itemize}
\end{definition}
\begin{theorem}
Let $R$ denote a function which near each point of positive charge $(p_2,p_3,0)$ is given by $R = \sqrt{(x_2 - p_2)^2 + (x_3 - p_3)^2 + y^2} + \mathcal O(x_2^2x_3^2)$ and $R = \mathcal O(1)$ uniformly as $\rho \to \infty$. Also define $\hat \rho := \sqrt{\rho^2+1}$. If $\mu \in (-2,1)$ and $\nu_1,\nu_2 \in (-1/2 - \sqrt{\gamma_0 +1/4}, -1/2 + \sqrt{\gamma_0 + 1/4})$ then for all $k \geq 0$ and $0<a<1$, 

\begin{equation}
\mathcal L_H: \psi^{\mu}R^{\nu_1}\hat \rho^{\nu_2}C^{k+2,a}_{\text{ie}}(\C\times\R^+; i\mathfrak{su}(E,H_0)) \rightarrow \psi^{\mu-2}R^{\nu_1-2}\hat \rho^{\nu_2 -2}C^{k,a}_{\text{ie}}(\C\times\R^+; i\mathfrak{su}(E,H_0))
\end{equation}
is an isomorphism.
\end{theorem}
\begin{proof}
The proof of the fact that the operator $\mathcal L_H$ is Fredholm of index zero between these space is sketched in \cite{MH2} where further references are given. In order to prove that this map is an isomorphism we therefore need to show that it is injective. Assume that there exists $s$ in the domain such that $\mathcal L_H s = 0$. Then it follows for example from the regularity theory developed in the later parts of \cite{MW2} that in the region $\rho < 16r_0$, $|s| \le C\psi R^{-1/2 + \sqrt{\gamma_0 + 1/4}}$ and in the region $\rho > 8r_0$, $|s| \le \psi\rho^{-1/2 - \sqrt{\gamma_0 + 1/4}}$. Since $\gamma_0$ is strictly greater than $2$, these bounds jastify an integration by parts using the formula \eqref{linearized operator} and therefore $\nabla_1 s = \nabla_2 s = \nabla _y s = 0$. Since $s$ has to vanish near infinity, we obtain that indeed $s = 0$ as wanted.
\end{proof}

\subsection{Improving the approximate solution}

Now that we have developed the linear analysis required, we can explain how to correct the approximate metric $H_0$ constructed in the previous section so that $\mathcal M(H_0) = \mathcal O(y^N)$ as $y\to 0$ for all $N>0$ as promised. The approximate metric constructed in the previous section satisfies $\mathcal M(H_0) = \mathcal O(\psi^{-1}R^{-1}\hat \rho^{-4})$. We want for all $j\ge -1$ to iteratively find sections $s_j\in i\mathfrak{su}(E, H_0)$ such that $H_0^{(j+1)} := H_0^{(j)}e^{s_j}$ satisfies $\mathcal M(H_0^{(j+1)}) = F_{j+1}\psi^{j+1}R^{j+1}\hat \rho^{-4} + \mathcal O(\psi^{j+2}R^{j+2}\hat \rho^{-4-\epsilon})$. Here $F_{j+1}$ is an everywhere bounded smooth section. Using equation \eqref{linearization} we see that it suffices to solve the equation 
\begin{equation}\label{local y}
\gamma(-s_j)\mathcal L_{H_0^{(j)}} s_j = - F_j\psi^jR^j\hat \rho^{-4} +  \mathcal O(\psi^{j+1-\epsilon}R^{j+1-\epsilon}\hat \rho^{-4-\epsilon}).
\end{equation}
First note that since $j \ge -1$, and we look for $s_j$ which is bounded, the equation forces $s_j = \mathcal O(y^{1-\epsilon})$ as $y\to 0$. We may not be able to set $\epsilon = 0$ because $1$ is an indicial root for $\psi$. Now, notice that if $s_j$ decays at least as fast as $\rho^{-\epsilon}$ as $\rho\to \infty$ then $\gamma(-s_j) = 1 + \mathcal O(y^{1-\epsilon}\hat \rho^{-\epsilon})$ and therefore $s_j$ must decay at precisely the rate $\rho^{-2}$ at infinity. Next, since $\mathcal L_{H_0^{(j)}}$ equals the normal operator $N(\mathcal H_0)$ to leading order, it suffices to solve the equation 
\begin{equation}
N(\mathcal L_{H_0})s_j = -  F_j\psi^jR^j\hat \rho^{-4}.
\end{equation}
As long as $j+2$ is not an indicial root for either $\psi$ or $R$ we can solve this equation near $y = 0$ so that $s_j = \mathcal O(\psi^{j+2}R^{j+2}\hat\rho^{-2})$. Otherwise if $j+2$ is an indicial root for either $\psi$ or $R$, then the solution will be of the form $s_j = \mathcal O(\psi^{j+2}(\log \psi)^iR^{j+2}(\log R)^k\hat\rho^{-2})$ for some finite numbers $i,k$. The possible appearance of these extra $\log$ terms is the reason why we needed to relax the condition on the error terms in equation \eqref{local y}. Finally, we take a Borel sum to obtain a section 
\begin{equation}
s \sim \sum s_j
\end{equation}
which satisfies $\mathcal M(H_0e^s) = \mathcal O(y^N)$ as $y\to 0$ for all $N>0$.

\section{The continuity method}

\subsection{Set up and openness}

The main idea is very similar to the deformation argument of \cite{MH2}. We start by considering the operator 
\begin{equation}\label{continuity operator}
N_t(s) := \text{Ad}(e^{s/2})\mathcal M(H) + ts = 0,
\end{equation}
where $H = H_0e^s$, $H_0$ is the approximate metric we have constructed above and the term $\text{Ad}(e^{s/2}): i\mathfrak{su}(E,H)\rightarrow i\mathfrak{su}(E,H_0)$ is a bundle isomorphism. Moreover, 
\begin{equation}
N_t(s): \hat{y}^{\mu}R^{\nu}\hat \rho^{-\beta}C^{k+2,a}_{\text{ie}}(\C\times\R^+; i\mathfrak{su}(E,H_0)) \rightarrow \hat{y}^{\mu-2}R^{\nu-2}\hat \rho^{-\beta}C^{k,a}_{\text{ie}}(\C\times\R^+; i\mathfrak{su}(E,H_0))
\end{equation}
is a smooth map depending smoothly on $t\in(0,1]$. Here \[\hat{y}:= \frac{y}{\sqrt{1+y^2}}.\] The continuity method amounts to showing that the operator  \eqref{continuity operator} has a solution $s$ for all $t\in[0,1]$. The fact that the operator $N_t(s)$ is not smooth between the above Banach spaces when $t = 0$ is not a problem as we show in the last part of this section. This failure of smoothness as $t\to 0$ happens precisely because due to the radial non-compactness of the domain, the deformation parameter $t$ changes the mapping properties of the linearized operator.  From this point onward we assume that $\beta = 3/2 +\epsilon$ for some small positive constant $\epsilon$. This is essential to prove that the linearization of $N_t(s)$ is invertible for $t\in(0,1]$. 

Let us first show that the set of $t\in[0,1]$ for which the operator $N_t(s)$ has a solution is non-empty. Define $s = \mathcal N(H_0)$. Clearly $s \in \psi^{\mu}R^{\nu}\hat \rho^{-\beta}C^{k+2,a}_{\text{ie}}(\C\times\R^+; i\mathfrak{su}(E,H_0))$ with $\beta = -4$ and therefore also with the required decay $3/2 + \epsilon$. Then it is straightforward to check that $N_1(-s) = 0$ and therefore the operator has a solution for  $t = 1$. 

Next we show that the set of points $t\in(0,1]$ is open. In order to do this we show that the linearization of the operator $N_t(s)$ is invertible between the appropriate spaces. The linearized operator is given by the formula 
\begin{equation}\label{linearization t}
\mathcal L_{t,s}(s') := \text{Ad}(e^{s/2})\mathcal L_H s' + ts'.
\end{equation}
\begin{proposition}
For $\epsilon >0$ chosen sufficiently small so that $1+\epsilon < -1/2 +\sqrt{\gamma_0 +1/4}$ and for $t\in(0,1]$, the linearized operator
\begin{equation}
\mathcal L_{t,s} : \psi^{\mu}R^{1+\epsilon}\hat \rho^{-(3/2+\epsilon)}C^{k+2,a}_{\text{ie}}(\C\times\R^+; i\mathfrak{su}(E,H_0)) \rightarrow \psi^{-\epsilon}R^{-1+\epsilon}\hat \rho^{-(3/2+\epsilon)}C^{k,a}_{\text{ie}}(\C\times\R^+; i\mathfrak{su}(E,H_0))
\end{equation}
is an isomorphism. 
\end{proposition}

To prove the proposition we will need the fact that we can re-write the linearized operator \eqref{linearization} as 
\begin{equation}\label{linear identity}
\mathcal L_H = \frac{1}{2}((\mathcal D_1^{\dagger_H})^{\star}\mathcal D_1^{\dagger_H} +(\mathcal D_2^{\dagger_H})^{\star}\mathcal D_2^{\dagger_H}) + (\mathcal D_3^{\dagger_H})^{\star}\mathcal D_3^{\dagger_H}.
\end{equation}
This can be proved using the K\"ahler identities $i[\Lambda, \mathcal D_i] = (\mathcal D_i^{\dagger_H})^{\star}$ and $i[\Lambda, \mathcal D_i^{\dagger_H}] = - (\mathcal D_i)^{\star}$.
\begin{proof}

It follows from the proof of proposition $6.4.$ in \cite{MH2} that the range of the operator $\mathcal L_{t,s}$ is dense. Therefore, surjectivity follows from the existence of a parametrix. Injectivity follows from the fact that its nullspace vanishes when the domain space is as above. In order to prove the latter, we consider $s'\in \psi^{2-\epsilon}R^{1+\epsilon}\hat \rho^{-(3/2+\epsilon)}C^{k+2,a}_{\text{ie}}(\C\times\R^+; i\mathfrak{su}(E,H_0))$ we compute 
\begin{equation}
\begin{split}
0 &= \int_{\R_3^+} \langle \mathcal L_{t,s} s' , \text{Ad}(e^{s/2})s'\rangle_{H_0}\\
&= \int_{\R_3^+} \langle \text{Ad}(e^{s/2})\mathcal L_H s' , \text{Ad}(e^{s/2})s'\rangle_{H_0} + t \langle s' , \text{Ad}(e^{s/2})s'\rangle_{H_0}\\
&= \int_{\R_3^+} \frac{1}{2}(|\mathcal D_1^{\dagger_H} s'|_{H}^2 + |\mathcal D_2^{\dagger_H} s'|_{H}^2) + |\mathcal D_3^{\dagger_H} s'|_{H}^2 + t |\text{Ad}(e^{s/4})s'|_{H_0}^2
\end{split}
\end{equation}
where in order to go from the second to the third line we have used the mapping property of the endomorphism $\text{Ad}(e^{s/2})$, the identity \eqref{linear identity} and we have integrated by parts. The integration by parts is justified by the choice of weights and the integral 
\begin{equation}
\int_{\R_3^+} |\text{Ad}(e^{s/4}) s'|_{H_0}^2
\end{equation}
is finite precisely because $s'$ decays at least as fast as $\rho^{-1 - \epsilon}$. This forces $s' = 0$ and therefore we are done.

In order to prove the existence of a parametrix, we only need to invert the leading part of the linearized operator in the region $\rho\to\infty$ since otherwise the parametrix is essentially the same as the one constructed in \cite{MH2}. Since the section $s$ is decaying as $\rho\to\infty$ we see that in the nilpotent case, the leading part of the linearized operator in the required region is given by the scalar operator 
\begin{equation*}
\mathcal L_N := \Delta -\frac{2}{y^2} - t.
\end{equation*}
We want to show that this operator is invertible as a map between 
\begin{equation*}
\mathcal L_N : \hat{y}^{\mu}\hat\rho^{-\beta}C^{k+2,a}_0 \rightarrow \hat{y}^{\mu-2}\hat\rho^{-\beta}C^{k,a}_0.
\end{equation*}
It is enough to show this for $\beta = 0$ since showing this for different $\beta$ follows from conjugating the operator with $\hat\rho^{\beta}$ which produces a compact error and then showing that the kernel of the new operator is still empty. The claim is equivalent to showing that 
\begin{equation*}
\mathcal L_N^1 := \hat y^{2-\mu}\circ \mathcal L_n \circ \hat y^{\mu}: C^{k+2,a}_0 \rightarrow C^{k,a}_0
\end{equation*}
is invertible. This will follow from the interior and boundary estimates of \cite{MH2} once we can show that given $g\in C^{k,a}_0$ we can find a solution to 
\begin{equation*}
\mathcal L_N^1 f = g
\end{equation*}
such that $\sup|f| < \infty$. Consider an exhaustion $\{K_n\}_{n\in\N}$ of $\R^2\times(0,\infty)$ by compact subsets with smooth boundary and solve $\mathcal L_N^1 f_n = g$ with Dirichlet boundary conditions on $\pa K_n$. Then a subsequence of the $f_n$ will converge to a solution $f$ on all of $\R^2\times\R_+$. We show that each $f_n$ is bounded from above and below by constant functions. We have that 
\begin{equation*}
\mathcal L_N^1 = \hat y^2\Delta + \frac{\mu \hat y^2}{y(y^2+1)}\pa_y + \frac{\mu(\mu -3y^2-1)}{y^2(y^2+1)^2}\hat y^2 - 2\frac{\hat y^2}{y^2} - t\hat y^2.
\end{equation*}
In particular, since $\mu\in(-1,2)$ the constant term of this operator is less than a constant negative number $-b$ everywhere on $\R^2\times\R_+$. Then, since $|g| \le M$ we know that 
\begin{equation*}
\begin{split}
\mathcal L_N^1 (f_n - M/b) &\ge 0 \\
\mathcal L_N^1 (f_n + M/b) &\le 0.
\end{split}
\end{equation*}
By the weak maximum principle and the fact that $f_n$ vanishes at the boundary $\pa K_n$, it follows that $|f_n| \le M/b$ for all $n$. Therefore we are done.

\end{proof}

\subsection{A priori estimates}\label{estimates}
We begin with a lemma. 
\begin{lemma}\label{main lemma}
Consider the problem on $\R^2\times \R_+$
\begin{equation}
\begin{cases}
\Delta u -tu &= f \\
u|_{y = 0} &= 0,
\end{cases}
\end{equation}
where $f$ is a smooth function which decays like $\rho^{-(4+2\beta)}$ near infinity, where $\beta$ is a positive constant and $t\in[0,1]$. When $t = 0$, the problem has a smooth solution which vanishes linearly and uniformly as $y\to 0$ and which decays quadratically and uniformly as $\rho$ goes to infinity. When $t > 0$, the problem has a smooth solution which vanishes linearly and uniformly as $y\to 0$ and such that there exists a constant only depending on $f$ and $t$ satisfying
\begin{equation}
|u| \le C_{f,t}\hat\rho^{-(2+2\beta')}
\end{equation}
for all $\beta' <\beta$.
\end{lemma}
\begin{proof}

This problem has an explicit solution of the form 
\begin{equation}\label{solution formula}
u = \int_{\R^3_+} f(x')G_t(x,x') \,dx',
\end{equation}
where $G_t(x,x')$ is Green's function for the operator $\Delta -t$ on the half space, given by the explicit formula 
\begin{equation}
G_t(x,x') := \frac{e^{-\sqrt{t}|x - x'|}}{4\pi|x - x'|} - \frac{e^{-\sqrt{t}|x-\bar x'|}}{4\pi|x - \bar x'|},
\end{equation}
with $x' = (x_2' , x_3' , y')$ and $\bar x' = (x_2' , x_3' , -y')$. 

Let us treat the case $t = 0$ first. Re-write $G_0(x,x')$ as
\begin{equation}
G(x,x') = \frac{yy'}{\pi|x - x'||x - \bar x'|(|x - x'| + |x - \bar x'|)},
\end{equation}
so that
\begin{equation}
u = y\int_{\R^3_+} \frac{f(x')y'}{\pi|x - x'||x - \bar x'|(|x - x'| + |x - \bar x'|)} \,dx'. 
\end{equation}
Our aim is to bound this integral. We split the integral into two regions which we treat separately. If $|x| = R$ then the first region will be $B_{R/2}(x)\cap \R_3^+$ and the second region will be its complement. In the first region, the condition on $f$ implies that $f(x') \le CR^{-(4+2\beta)}$ and it also holds that
\begin{equation}
\frac{y'}{|x - x'| + |x - \bar x'|} \le1. 
\end{equation}
Finally, the integral 
\begin{equation}
\int_{B_{R/2}(x)\cap \R_3^+} \frac{1}{\pi|x - x'||x - \bar x'|} \,dx'
\end{equation}
can be bounded above by $CR\log R$ where $C$ does not depend on $x$. Therefore, in the first region it holds that 
\begin{equation}
|u| \le CyR^{-(3+2\beta')}
\end{equation}
for any $\beta'<\beta$. We cannot have equality because of the $\log$ term in the bound, which cannot be avoided. In the second region each of the three factors in the denominator is bounded below by $R/2$ and therefore 
\begin{equation}
\begin{split}
|u| &\le C\frac{y}{R^3}\int_{\R_3^+} |f(x')y'|\,dx' \\
&\le C\frac{y}{R^3},
\end{split}
\end{equation}
since $|f(x')y'| = \mathcal O (\hat\rho^{-(3+2\beta)})$ and thus integrable. 

Now consider the case $t>0$. We will split the integral in equation \eqref{solution formula} into the same two regions as we did in the $t = 0$ case. In the first region we use the inequality 
\begin{equation}
\begin{split}
0 \le G_t(x,x') &=  \frac{e^{-\sqrt{t}|x - x'|}}{4\pi|x - x'|} - \frac{e^{-\sqrt{t}|x-\bar x'|}}{4\pi|x - \bar x'|}\\
&\le \frac{1}{4\pi|x - x'|} - \frac{1}{4\pi|x - \bar x'|}.
\end{split}
\end{equation}
One can derive this inequality as follows. Consider the function $f(x) = a^{-1}e^{-ax} - b^{-1}e^{-bx}$ where $0<a<b$. Its derivative is equal to $e^{-bx} - e^{-ax} \le 0$ and therefore its maximal value is at $x = 0$ and its infimum when $x\to \infty$. Therefore, the same bound holds in this region as in the $t = 0$ case and the constant in fact does not depend on $t$. 

In the second region, we use the inequality 
\begin{equation}
\begin{split}
0 \le G_t(x,x') &=  \frac{e^{-\sqrt{t}|x - x'|}}{4\pi|x - x'|} - \frac{e^{-\sqrt{t}|x-\bar x'|}}{4\pi|x - \bar x'|}\\
&\le \frac{C_{\beta}}{t^\beta}\left(\frac{1}{|x - x'|^{1+2\beta}} - \frac{1}{|x - \bar x'|^{1+2\beta}}\right).
\end{split}
\end{equation}
Again, to prove this, consider the function $f(x) = C_{\beta}(a^{-1}(ax)^{-2\beta} - b^{-1}(bx)^{-2\beta})  -  (a^{-1}e^{-ax} - b^{-1}e^{-bx})$. We may assume that $1<a<b$ since we only care about $|x| = R$ going to infinity. Under this assumption it is easy to check that one can choose $C_{\beta}$ so that the derivative of this function is negative everywhere and $f(0) = 0$ therefore we are done. 

Now, if $2N$ is the smallest even positive integer greater than $1+2\beta$ then it is easy to check that 
\begin{equation}
\begin{split}
\frac{1}{|x - x'|^{1+2\beta}} - \frac{1}{|x - \bar x'|^{1+2\beta}} &\le \frac{|x - x'|^{2N} - |x - \bar x'|^{2N}}{|x - x'|^{2N}|x - \bar x'|^{1 + 2\beta}}\\
& = \frac{4yy'\sum\limits_{i+j\le 2N -1}|x - x'|^i|x - \bar x'|^j}{|x - x'|^{2N}|x - \bar x'|^{1 + 2\beta}}.
\end{split}
\end{equation}
Finally, as we did in the $t = 0$ case, the $y'$ can be absorbed by $f$ and therefore the same argument will give a decay rate of $\rho^{-(2 +2\beta)}$. Since $\beta' < \beta$ we are done. 
\end{proof}
\begin{proposition}
If $s$ is a Hermitian endomorphism which satisfies $N_t(s) = 0$ for $t\in[0,1]$ then there exists $C$ only depending on $H_0$ such that 
\begin{equation}
\sup |\hat \rho s| \le C. 
\end{equation}
If we only restrict to $t\in(0,1]$ then there exists $C$ only depending on $H_0$ such that 
\begin{equation}
\sup |\hat \rho ^{3/2 + \epsilon} s| \le Ct^{-(3/2+ \epsilon)}.
\end{equation}
\end{proposition}
\begin{proof}
Using equation \eqref{inner product identity} we obtain the identity 
\begin{equation}
-\Delta|s|^2 + |\sqrt{\gamma(-s)}\nabla s|^2 + t|s|^2 + \langle\Omega_{H_0} , s\rangle = 0.
\end{equation}
Let us first prove the assertion for $t\in[0,1]$. We initially assume that $|s| \le C_{s}\hat\rho^{-1}$ for some constant depending on $s$. Then we get the differential inequality 
\begin{equation}
-\Delta|s|^2 \le -\langle\Omega_{H_0} , s\rangle \le MC_{s}\hat\rho^{-5},
\end{equation}
where $M$ only depends on $H_0$. From lemma \eqref{main lemma}, we can find $u$ vanishing linearly in $y$ and quadratically as $\rho \to \infty$ which solves $-\Delta u = - \hat\rho^{-5}$. Therefore we get that $\Delta(|s|^2 - MC_{s}u) \ge 0$. By the weak maximum principle and the fact that both $s$ and $u$ vanish at all boundaries and near infinity that $|s| \le (MC_{s})^{1/2}\hat\rho^{-1}$. Therefore we have improved the constant from $C_{s}$ to $(MC_{s})^{1/2}$. Repeating this process enough times, it follows that in fact $|s| \le (M +1)\hat\rho^{-1}$.

Now, restricting $t\in(0,1]$, we initially assume that $|s| \le C_{s,t}\hat\rho^{-(3/2+\epsilon)}$. Then, we have the differential inequality 
\begin{equation}
-(\Delta - t)|s|^2 \le MC_{s,t}\hat\rho^{-(11/2 + \epsilon)}.
\end{equation}
Using lemma \eqref{main lemma} we can find $u$ solving $(\Delta - t)u = -\hat\rho^{-(11/2 + \epsilon)}$, vanishing linearly in $y$ and since $\beta = 3/4 + \epsilon/2$ here, decay like $\hat\rho^{-(2 + 3/2)}$. Therefore, 

\begin{equation}
(\Delta - t)(|s|^2 - MC_{s,t}t^{-\beta}u) \ge 0.
\end{equation}
By the weak maximum principle and the boundary behavior of $s$ and $u$ we obtain the inequality $|s| \le (MC_{s,t})^{1/2}t^{-\beta} \hat\rho^{-7/4} \le (MC_{s,t})^{1/2}t^{-\beta} \hat\rho^{-(3/2 + \epsilon)}$. Therefore we obtain the bound $|s| \le (M+1)t^{-2\beta}\hat\rho^{-(3/2 + \epsilon)}$.
\end{proof}
\begin{theorem}\label{main}
Suppose $N_t(s) = 0$, where $H_0$ is the approximate metric constructed above. Then the following estimates hold with $C$ a universal constant, for all non-negative integers $k$, $0<a<1$ and $\epsilon$ appropriately small.
\begin{itemize}
\item For $t\in(0,1]$ , $\|s\|_{\mathcal A_{k,a,1/2 + \epsilon}} \le Ct^{-(3/2 + \epsilon)}$
\item For $t\in[0,1]$ , $\|s\|_{\mathcal A_{k,a,0}} \le C$
\end{itemize}
where 
\begin{equation}
\mathcal A_{k,a,\beta}:= \psi^{2-\epsilon}R^{1+\epsilon}\hat \rho^{-(1 + \beta)}C^{k,a}_{\text{ie}}(\C\times\R^+; i\mathfrak{su}(E,H_0))
\end{equation}
\end{theorem}
\begin{remark}
We will not give the proof since it is identical to the main theorem proved in section $8$ of \cite{MH2}. The interior estimates as well as the estimates as $\rho \to \infty$ go back to the paper of Bando and Siu \cite{bando} and Hildebrandt \cite{hildebrandt} and depend on the particular form of the operator $N_t(s)$. The boundary estimates sensitive to the singular nature of the operator were proved in \cite{MH2} by adapting the above methods appropriately and making use of the scaling properties of the operator $N_t(s)$. 
\end{remark}

\subsection{Existence}

Finally, we prove existence. If we have a sequence of solutions $s_j$ of $N_{t_j}(s_j) = 0$ such that $\{t_j\}$ is decreasing and converging to $t >0$ then the first inequality in theorem \eqref{main} implies that these sections are uniformly bounded in $\mathcal A_{k,a,\epsilon}$ and therefore there exists $0 < \epsilon ' < \epsilon$ and $0 <a' <a$ such that a subsequence of this sequence converges in $\mathcal A_{k,a',\epsilon'}$. Now, the limit $s_t := \lim s_{t_j}$ solves the equation $N_t(s) = 0$ and the regularity theory for the Extended Bogomolny Equations with knots as developed in \cite{MW2} implies that $s$ is in fact in  $\mathcal A_{k,a,\epsilon}$. Combining this with the linear analysis in the previous section we have shown that the equation $N_t(s) = 0$ has a solution in  $\mathcal A_{k,a,\epsilon}$ for all $t\in (0,1]$. Finally, the second a priori estimate in theorem \eqref{main} implies that if $s_j$ is a sequence of solutions to $N_{t_j}(s_j) = 0$ with $t_j \to 0$ then this sequence is uniformly bounded in $\mathcal A_{k,a,0}$ and therefore has a convergent subsequence in $\mathcal A_{k,a',-\epsilon}$. This subsequence converges to $s$ which solves the equation $N_0{s} = 0$ which is what we wanted. From this point onward the regularity theory of \cite{MW2} can be used to describe the exact asymptotics of $s$ and therefore of the solution metric. 

\section{Uniqueness}

Following \cite{MH2} we prove uniqueness of solutions using the convexity of the Donaldson functional. For any two metrics $K$ and $H$ such that $H = Ke^s$ such that $\text{Tr}(s) = 0$, write 
\begin{equation*}
\mathcal M_{K,H}:= \frac{i}{2} \Lambda\left( [\mathcal D_1 , \mathcal D_1^{\dagger_H}] + [\mathcal D_2 , \mathcal D_2^{\dagger_H}]\right) + [\mathcal D_3 , \mathcal D_3^{\dagger_H}],
\end{equation*}
where the conjugate $\mathcal D_i^{\dagger}$ is taken with respect to $H$. The subscript $K$ emphasises that when $H = K$ then we are considering $\mathcal M = 0$ as an equation for $s$. The analog of the Donaldson functional for the extended Bogomolny equations takes the form 
\begin{equation}
\mathcal F(H,K):= \int_0^1\int_{\R^2\times\R^+} \langle s, \mathcal M_{Ke^{us}, K}\rangle_K \,dz\wedge\,dy\wedge\,du.
\end{equation}
Writing $H_t = Ke^{ts}$ we have the formulas 
\begin{equation}\label{derivatives of Donaldson functional}
\begin{split}
\frac{d}{dt}\mathcal F(H_t,K) &= \int_{\R^2\times\R^+} \text{Tr}(\mathcal M_{H_t,K}s)\,dz\wedge\,dy \\
\frac{d^2}{dt^2} \mathcal F(H_t,K) &= \sum\limits_{i=1}^3 \int_{\R^2\times\R^+} |\mathcal D_i s|^2 + \int_{\R^2\times\R^+} \pa_{\bar z}\text{Tr}(\mathcal D_1^{\dagger}s\wedge s) + \int_{\R^2\times\R^+} \pa_y\text{Tr}(\mathcal D_3^{\dagger}s\wedge s).
\end{split}
\end{equation}
Using these expressions we show that the map from holomorphic data to solutions to the extended Bogomolny equations is injective. 
\begin{theorem}
Given a triple $(P(z), Q(z), R(z))$ of polynomials satisfying the degree condition, suppose that $H,K$ are two solutions to the extended Bogomolny equations corresponding to this holomorphic data. Then $H = K$. 
\end{theorem}
\begin{proof}
Write $H = Ke^s$ so that $H_t = Ke^{ts}$. The indicial computations for $\mathcal L$ imply that the order of vanishing of $s$ is strictly greater than $1$ both on the boundary and at radial infinity. Therefore, the boundary terms in \eqref{derivatives of Donaldson functional} vanish. The condition $\mathcal D_1 s = \mathcal D_3 s = 0$ forces $s$ to be holomorphic, which together with its vanishing behavior forces it to be zero. Therefore, $\frac{d^2}{dt^2} \mathcal F(H_t,K) > 0$ if $s\neq 0$. However, since $\mathcal F(H_0,K) = \mathcal F(H,K) = \frac{d}{dt}\mathcal F(H_0,K) = 0$, we obtain that $\mathcal F(H_t,K) =0$ and therefore $H = K$. 
\end{proof}

\begin{remark}\label{Taubes' solutions}
Having this uniqueness result, we can identify Taubes' instanton solutions with the holomorphic data $(P(z), Q(z), R(z)) = (z^k, a_0+...+a_{p-1}z^{p-1}, z^p)$ and therefore recover the $\C^{p-1}\times\C^{\star}$ moduli space of gauge inequivalent instanton model solutions that he described. 
\end{remark}

\begin{remark}
Contingent on generalizing theorem \ref{diagonal metric solution} to the $SU(n+1)$ case, the methods of this paper can be used in a straightforward way to construct multi-instanton solutions for the extended Bogomolny equations with group $SU(n+1)$ for $n>1$ since model solutions in this case have been constructed in \cite{Mikhaylov}.
\end{remark}

\bibliography{bibliography}
\bibliographystyle{alpha}

\end{document}